\documentclass[11pt]{article}
\usepackage{mathpazo}
\usepackage[T1]{fontenc}
\usepackage[margin=1in]{geometry}
\usepackage{amsmath,amssymb,amsfonts,amsthm}
\usepackage{mathtools}
\usepackage{bm}
\usepackage{bbm}

\usepackage[dvipsnames]{xcolor}
\usepackage{xspace}
\usepackage{microtype}
\usepackage{csquotes}
\usepackage{graphicx}
\usepackage{subcaption}
\usepackage{adjustbox}
\usepackage{float}
\usepackage{enumitem}
\setlist{nosep,topsep=0pt,leftmargin=*}
\usepackage{url}

\usepackage{alphabeta}

\usepackage{todonotes}
\usepackage{authblk}

\usepackage[ruled,vlined,linesnumbered,longend]{algorithm2e}
\SetKw{KwInit}{Initialization:}
\SetKwComment{Comment}{/* }{ */}

\definecolor{myRed}{rgb}{0.82,0.13,0.56}
\definecolor{myBlue}{RGB}{13,55,174}
\usepackage{hyperref}
\hypersetup{
  colorlinks=true,
  citecolor=myRed,
  linkcolor=myBlue,
  urlcolor=PineGreen
}
\usepackage[noabbrev,capitalise]{cleveref}
\Crefname{algocf}{Algorithm}{Algorithms}

\usepackage{thmtools}
\usepackage{thm-restate}

\usepackage{eqparbox}

\newcommand{\lp}{\left}
\newcommand{\rp}{\right}
\newcommand{\E}[2][]{\mathbb{E}_{#1}\!\left[#2\right]}
\newcommand{\Prob}[2][]{\mathbb{P}_{#1}\!\left[#2\right]}
\newcommand{\Exp}[1]{\exp\lp(#1\rp)}
\newcommand{\prn}[1]{\left( #1 \right)}
\newcommand{\abs}[1]{\left|#1\right|}

\newcommand{\floor}[1]{\left\lfloor #1 \right\rfloor}

\newcommand{\e}[1]{\texttt{empty}(#1)}

\newcommand{\ALG}{\textsc{Alg}}
\newcommand{\OPT}{\textsc{Opt}}

\theoremstyle{plain}
\newtheorem{theorem}{Theorem}[section]
\newtheorem{proposition}[theorem]{Proposition}
\newtheorem{lemma}[theorem]{Lemma}

\newtheorem{corollary}[theorem]{Corollary}
\theoremstyle{definition}
\newtheorem{definition}[theorem]{Definition}

\newtheorem{remark}[theorem]{Remark}
\theoremstyle{remark}

\author[1]{Andreas Kalavas\thanks{The majority of this work was carried out while the author was an intern at the Archimedes Research Unit.}}
\author[2,3]{Charalampos Platanos}
\author[2,3]{Thanos Tolias}

\affil[1]{Carnegie Mellon University, Pittsburgh, USA}
\affil[2]{National Technical University of Athens, Greece}
\affil[3]{Archimedes RU, Athena RC, Greece}

\affil[ ]{\texttt{{
\href{mailto:akalavas@andrew.cmu.edu}{akalavas@andrew.cmu.edu} \quad
\href{mailto:harrisplat@gmail.com}{harrisplat@gmail.com} \quad
\href{mailto:thanostolias@mail.ntua.gr}{thanostolias@mail.ntua.gr}}}}

\title{A Polylogarithmic Competitive Algorithm for Stochastic Online Sorting and TSP}
\date{}

\begin{document}

\maketitle

\begin{abstract}
In \emph{Online Sorting}, an array of $n$ initially empty cells is given. At each time step $t$, an element $x_t \in [0,1]$ arrives and must be irrevocably placed in an empty cell without knowledge of future arrivals. We aim to minimize the sum of absolute differences between pairs of elements placed in consecutive array cells, seeking an online placement strategy that results in a final array close to a sorted one.   
An interesting multidimensional generalization, referred to as the \emph{Online Traveling Salesperson Problem}, arises when the request sequence consists of points in the $d$-dimensional unit cube and the objective is to minimize the sum of Euclidean distances between points in consecutive cells. 
Motivated by the recent work of (Abrahamsen, Bercea, Beretta, Klausen and Kozma; ESA 2024), we consider the \emph{stochastic version} of Online Sorting (\textit{resp.} Online TSP), where each element (\textit{resp.} point) $x_t$ is an i.i.d. sample from the uniform distribution on $[0, 1]$ (\textit{resp.} $[0,1]^d$). By carefully decomposing the request sequence into a hierarchy of balls-into-bins instances, where the balls to bins ratio is large enough so that bin occupancy is sharply concentrated around its mean and small enough so that we can efficiently deal with the elements placed in the same bin, we obtain an online algorithm that approximates the optimal cost within a factor of $O(\log^2 n)$ with high probability. 
Our result comprises an exponential improvement over the previously best known competitive ratio of $\tilde{O}(n^{1/4})$ for Stochastic Online Sorting due to (Abrahamsen et al.; ESA 2024) and $O(\sqrt{n})$ for (adversarial) Online TSP due to (Bertram, ESA 2025). 
\end{abstract}

\begingroup
\renewcommand\thefootnote{}\footnotetext{
This work has been partially supported by project MIS~5154714 of the National Recovery and Resilience Plan Greece~2.0, funded by the European Union under the NextGenerationEU Program.

The authors would like to thank Dimitris Fotakis for many valuable discussions and insightful comments on this paper. They are also grateful to Marina Kontalexi for her helpful discussions.
}
\addtocounter{footnote}{0}
\endgroup

\newpage

\section{Introduction}

In the \emph{Online Sorting Problem} we are given a sequence of $n$ real numbers $x_1, x_2, \ldots, x_n \in [0,1]$, revealed one by one in an online fashion. An array $A$ of length $n$ is initially empty. Denote by $A[1], A[2], \ldots, A[n]$ its cells. Upon the arrival of each element $x_j$, the algorithm must immediately and irrevocably assign it to an empty cell of $A$. After all $n$ elements have been placed, the cost is defined as the total variation between consecutive entries \textit{i.e.}, $\text{Cost}(A) = \sum_{i=1}^{n-1}|{A[i+1] - A[i]}|$. The objective is to minimize this cost. The problem was introduced by Aamand, Abrahamsen, Beretta, and Kleist~\cite{abrSODA} as a technical tool for proving lower bounds in online strip packing, bin packing, and perimeter packing. They studied the adversarial setting, where the sequence is chosen by an adaptive adversary, and designed an $O(\sqrt{n})$-competitive algorithm along with a matching lower bound for deterministic algorithms. Later, Abrahamsen, Bercea, Beretta, Klausen, and Kozma~\cite{AbrESA} showed that even randomized algorithms cannot asymptotically improve this guarantee in the worst case.

\subparagraph*{Stochastic Online Sorting.} Motivated by the large $\Omega(\sqrt{n})$ lower bound in the adversarial setting, Abrahamsen et al.~\cite{AbrESA} introduced \emph{Stochastic Online Sorting}. There, the input elements $x_1, \dots, x_n$ are drawn \textit{i.i.d.}~from the uniform distribution $\mathcal{U}(0, 1)$ and the algorithm seeks to minimize the cost incurred. Note that in the stochastic online sorting problem it is known that the cost of the offline optimal solution $\OPT$ is approximately equal to one for large enough $n$. Throughout this paper, we say that an algorithm for stochastic online sorting is $c$-competitive if it achieves cost at most $c\cdot \OPT$ with high probability\footnote{Throughout this work, “with high probability” means with probability at least $1 - 1/n$.}. Abrahamsen et al.~\cite{AbrESA}, designed a $\widetilde{O}(n^{1/4})$-competitive algorithm, demonstrating that probabilistic assumptions can lead to significant improvements over the adversarial baseline. 

\subparagraph*{Online TSP.} 
Abrahamsen et al.~\cite{AbrESA} further generalized the problem by increasing its dimensionality, introducing the \emph{Online Traveling Salesperson Problem}. Here the input elements $x_1, \dots, x_n$ are points in $[0,1]^d$, and the objective is to minimize the total Euclidean distance between consecutive points in the array. In the adversarial setting, they gave a dimension-dependent algorithm with competitive ratio $\sqrt{d}\cdot 2^d \cdot O(\sqrt{n \log n})$. Shortly after, Bertram~\cite{bertESA} showed that dimension-independent bounds are achievable in the adversarial model presenting an $O(\sqrt{n})$-competitive algorithm. To the best of our knowledge, the competitive ratio of the stochastic version of this problem has not been studied.

\subsection{Our Contributions} 
In this work, we present a unified algorithmic framework, \Cref{theAlg}, that applies to both stochastic online sorting and its multidimensional generalization, stochastic online TSP. Our framework carefully decomposes the input sequence into a series of balls-into-bins instances to achieve an $O(\log^{2} n)$ approximation with high probability in both settings. We formally state our two main theorems below.

\begin{restatable}{theorem}{sorting} \label{thm:1}
    \Cref{theAlg} achieves a cost of at most $O(\log^2 n)$ with high probability for the stochastic online sorting problem.
\end{restatable}

\begin{restatable}{theorem}{tsp} \label{thm:2}    \Cref{theAlg} achieves a cost of at most $O(\log^2 n) \cdot \OPT$ with high probability for the stochastic online TSP problem.
\end{restatable}

\subparagraph*{Independent Work.}

In a recent parallel and independent work, Hu~\cite{hu2025} studied Stochastic Online Sorting and obtained a deterministic algorithm with expected cost at most $\log n \cdot 2^{O(\log^\ast n)}$ and an elegant lower bound of $\Omega(\log n)$ on the expected cost of any randomized algorithm. Regarding high-probability bounds, which we aim to derive in this work, Hu \cite[Theorem~1.3]{hu2025} presented an $O(\log^{20} n)$-competitive algorithm.
 
\subsection{Technical Overview}

\subparagraph*{An $O(\log^2 n)$-Competitive Algorithm for Stochastic Online Sorting.}

The first main idea that Abrahamsen et al.~\cite{AbrESA} introduced is to divide the array into buckets / bins, exploit the randomness of the input in order to assign each element to a bucket, and solve the problem separately in each bucket. Suppose that each bucket has $C$ elements; thus, we have $n/C$ buckets, and the bucket $i$ contains elements in $((i-1)\frac{C}{n}, i\frac{C}{n}]$. Then, assuming that each bucket will receive exactly the expected number of elements, \textit{i.e.} $C$ elements, the cost of the algorithm is approximately $\ALG = \text{number of buckets}\times \text{cost inside bucket} = O(\sqrt{C})$, assuming that the $O(\sqrt{n})$ adversarial algorithm of Aamand et al.~\cite{abrSODA} was employed in each bucket. 

Hence, smaller buckets reduce cost but increase the chance of imbalance, while larger buckets stabilize the load at the expense of higher internal cost. This trade-off is central to their analysis. To handle the imbalances, a part of the array in the end is reserved, also called backyard, and its size depends on how large these imbalances are. Since elements inserted in the backyard are arbitrary, the algorithm incurs a cost equal to the square root of its size. Thus, the previous trade-off translates into a trade-off between the size of the backyard and $C$: the larger $C$ is, the smaller the backyard and vice versa. Their first algorithm balances this trade-off and achieves $\widetilde{O}(n^{1/3})$ complexity. Recursively using their algorithm inside the buckets, they improve their competitive ratio to $\widetilde{O}(n^{1/4})$.

Our approach to Stochastic Online Sorting can be naturally described through the lens of the balls-into-bins paradigm. More specifically, the algorithm proceeds by maintaining a partitioning of a certain part of $A$ into equally sized \emph{bins}, while the arriving elements can be regarded as \emph{balls} placed into bins. Each bin is associated with a certain subinterval of $[0,1]$, which determines if a new ball / element is placed into the particular bin. The subintervals associated with the bins form a partitioning of $[0,1]$ and are of equal length. Hence, since the elements are \textit{i.i.d.} samples from the uniform distribution on $[0,1]$, each new ball is placed uniformly at random (and independently of the other balls) in each bin. 

Interpreting the $\widetilde{O}(n^{1/3})$-competitive algorithm of Abrahamsen et al. \cite{AbrESA} through the balls-into-bins framework above, we realize that the main limitation of their approach is that they consider a single balls-into-bins instance, which has very large bins and is fixed at the beginning of the algorithm. Therefore, the backyard of $A$ must be large enough, in order to accommodate the significant imbalances in the bins' occupancy. 

Our key new technical insight is that we can decompose an instance of Stochastic Online Sorting into a hierarchy of balls-into-bins instances (\textit{a.k.a.} \emph{phases}) of geometrically decreasing size. Crucially, we first distinguish between \emph{buckets} and \emph{bins}. A bucket is a contiguous set of cells corresponding to a single interval of the domain, while a bin is a collection of buckets that together cover the same interval but need not be contiguous in the array.  It is important, that we select the bucket sizes dynamically, when we proceed from one phase to the next, in order to create same sized bins to deal with any imbalances (\textit{i.e.}, empty cells in some buckets) left from the previous phase. 

Our algorithm (see also Section~\ref{sec:alg} and Algorithm~\ref{theAlg} for the details) maintains that all bins in the same phase are of equal size (and have equal probability of receiving a new element) and all bins created by the algorithm have size $\Theta(\log^2 n)$ (though the bucket size may slightly vary between phases). In the first phase, we consider the first $n/2$ cells of $A$, which we partition into $n/(2\log^2 n)$ buckets (we also partition $[0,1]$ into the same number of subintervals with equal length). The current phase ends when the first of its buckets becomes full. A standard concentration bound (Lemma~\ref{lem:balls-into-bins}) shows that \textit{w.h.p.} before the first phase ends: \textit{(i)} at least $(1 - o(1))\frac{n}{2}$ elements are successfully placed into its bins; and \textit{(ii)} every bin has received at least $(1 - o(1))\log^2 n$ elements. Using the (adversarial) $O(\sqrt{n})$-competitive algorithm of Aamand et al.~\cite{abrSODA} to deal with the exact placement of the elements in the same bin we get that \textit{w.h.p.} the total cost of the algorithm during its first phase is $O(\log n)$. 

In the second phase, we consider the next $n/4$ cells of $A$ which are partitioned into $n/(4\log^2 n)$ buckets (as before, we partition $[0,1]$ into the same number of subintervals with equal length). The bin size is slightly larger than in the first phase (but again at most $(1+o(1))\log^2 n$), because in the bins we also include cells from buckets left empty in the first half of $A$ from the previous phase. As above, \textit{w.h.p.} before the second phase ends: \textit{(i)} at least $(1 - o(1))\frac{n}{4}$ elements are successfully placed into its bins; \textit{(ii)} every bin has received at least $(1 - o(1))\log^2 n$ elements; and \textit{(iii)} the total cost of the algorithm for its second phase is $O(\log n)$. Moreover, we show that due to property \textit{(ii)} (combined for the first and the second phase), \textit{w.h.p.} before the second phase ends, all cells in the first half of $A$ (\textit{i.e.}, the cells considered in the first phase) are full. So, we maintain the invariant that any imbalances in bucket occupancy left from the first (\textit{resp.} any) phase do not carry over to the phases after the second (\textit{resp.} next one). 

Using the steps and the invariants (and carefully defining the exact quantities hidden in the $o(1)$ notation) above, our algorithm proceeds from one phase to the next, for $O(\log n)$ phases, until we are left with a single bucket of size $\Theta(\log^2 n)$. The total cost is dominated by the algorithm's total cost for the different phases and is $O(\log^2 n)$  \textit{w.h.p.}.

\subparagraph*{Extension to Stochastic Online TSP.} 
When extending to higher dimensions, three challenges arise. First, unlike in the one-dimensional setting, there is no simple closed-form expression for $\OPT$, and we must instead rely on getting a good estimate of $\OPT$ using concentration bounds. Second,  the domain must be partitioned into blocks that both \textit{(i)} have equal probability mass, to preserve the balls-into-bins property, and \textit{(ii)} contain points that are sufficiently close so that the intra-block cost is negligible compared to $\OPT$. Third, there is no obvious ordering of the blocks that guarantees low inter-block cost. To overcome these difficulties, we exploit properties of the uniform distribution, which allows us to partition the space into hyperboxes of similar geometry. We then define a tour that visits the input points block by block, with consecutive blocks chosen to be adjacent. The ordering of blocks is inspired by space-filling curves, a tool that has been widely used in the study of universal TSP (see \cite{space1, space2, space3, chrissgour}). We prove that this structured tour is within a constant factor of the optimal TSP tour \textit{w.h.p.}, and we adapt our algorithmic framework to approximate it using Bertram's algorithm~\cite{bertESA}. As a result, we retain the $O(\log^2 n)$ guarantee in higher dimensions.

\subsection{Other Related Work}

\subparagraph*{Online Sorting with Larger Arrays.} Another interesting variant of the \emph{Online Sorting Problem} is one where the size of the array $m$ is longer than the input sequence, \textit{i.e.}, $m > n$. This version was introduced by Aamand et al.~\cite{abrSODA}, who designed a deterministic $2^{\sqrt{\log n}\sqrt{\log\log n + \log(1/\epsilon)}}$-competitive algorithm when $m = (1+\epsilon)n$. They complemented this with a lower bound, showing that every deterministic algorithm with $m = \gamma n$ is at least $1/\gamma \cdot \Omega(\log n / \log \log n)$-competitive. Later, Azar et al.~\cite{azar2025} and, independently and concurrently, Nirjhor et al.~\cite{nirjhor2025} improved the upper bound, with the former nearly resolving this variant.

\subparagraph*{Hashing Schemes.}
The connection between stochastic online sorting and hashing was already observed by Abrahamsen et al.~\cite{AbrESA}. Two particularly relevant examples are \emph{Filter Hashing}, introduced by Fotakis et al.~\cite{fotakis2003}, and \emph{Transactional Multi-Writer Cuckoo Hashing}, proposed by Kuszmaul~\cite{kuszmaul2016}. Both hashing schemes employ a multi-layered structure where each subsequent layer is tasked with handling the imbalances of the previous layers.

\subparagraph*{Variants of Online TSP.}
Online TSP has also been studied in different contexts. A notable line of work interprets it as a scheduling problem: points (or requests) appear online in a metric space, and the objective is to minimize the time until all points have been visited~\cite{ausiello99}. In contrast, our setting focuses on minimizing the total length of the tour, rather than the completion time.

\section{General Algorithmic 
Framework} \label{sec:alg}

Before we present our algorithm, we introduce some notation. For an array $A$ we denote by $A[s:t]$ the subarray of $A$ starting at $s$ and ending at $t$. We also adopt a slight abuse of notation by using $A$ to refer both to the array itself and to its length, with the intended meaning being clear from the context. We also define by $[k]$ the set $\{1, 2, \ldots, k\}$. We call a subarray of a subarray of $A$ a \emph{sub-subarray} (or \emph{bucket}).

\subsection{Algorithm Description}
We present an online sorting algorithm for the general setting in which we are given an array $A$ of length $n$ and a distribution $\mathcal{D}$ over a domain $\mathcal{S}$, with sample access provided by \texttt{receive\_sample}.  We also assume access to the following  subroutines.
\subparagraph*{Subroutines.}
\begin{itemize}
    \item \underline{$\texttt{DomainPartitioning}\lp(\mathcal{D}, \mathcal{S}, \ell\rp)$:} This procedure takes as input a distribution $\mathcal{D}$ over a domain $\mathcal{S}$ and a positive integer $\ell$.  
    It outputs subsets $
    \bigl\{ T^{(i)}_j \subseteq \mathcal{S} : i \in [\ell+1],\; j \in [2^{\ell-i+1}] \bigr\},$ collectively denoted by $\mathcal{T}$.  
    The collection $\mathcal{T}$ satisfies, for every level $i \in [\ell+1]$:
    \begin{enumerate}
      \item \textit{Covering:} $\bigcup_{j \in [2^{\ell-i+1}]} T^{(i)}_j = \mathcal{S}.$
      \item \textit{Disjointness:} The sets $\{T^{(i)}_j\}_{j \in [2^{\ell-i+1}]}$ are pairwise disjoint.
      \item \textit{Equal Mass:} $\Prob[x \sim \mathcal{D}]{x \in T^{(i)}_j}
        \text{ is the same for all } j \in [2^{\ell-i+1}].$
      \item \textit{Laminarity:} $\mathcal{T}$ forms a \emph{laminar family}---every pair of sets in $\mathcal{T}$ is either disjoint or one is contained in the other.
    \end{enumerate}
    Thus $\mathcal{T}$ defines a hierarchical, binary-tree partition of $\mathcal{S}$, consisting of $2^{\ell+1}-1$ subsets in total.
    \item \underline{$\texttt{InBucketPlacement}(x, a)$:} This function takes as input an element $x \in \mathcal{S}$ and places it into an empty cell of sub-subarray $a$ of $A$. 
    \item \underline{$\texttt{index}(\mathcal{T}, i, x)$:} This function takes as input a set $\mathcal{T}$ outputted by $\texttt{DomainPartitioning}$, a positive integer $i$ and an input element $x$ and outputs the unique index $k$ such that $x \in T^{(i)}_k$.
    \item \underline{\texttt{empty}($a$):} This function takes as input a (sub-)subarray $a$ of $A$ and returns the number of empty cells in $a$ at the current point of execution.
\end{itemize}

We also introduce the following definition: 

\begin{definition}
For a subarray $A_i$, let $N_i \coloneqq \texttt{empty}(A_i)$ at the moment when the algorithm transitions from phase $i$ to phase $i+1$, i.e., when some bucket of $A_i$ becomes full and triggers an overflow.
\end{definition}

The complete pseudocode is given in \Cref{theAlgFull}; 
a more high-level version is presented in \cref{theAlg}. We now describe the algorithm step by step.

\begin{algorithm}[ht]
    \caption{General Algorithmic Framework}\label{theAlg} 
    \KwData{Array $A[1:n]$, Distribution $\mathcal{D}$ over a domain $\mathcal{S}$}
    \KwResult{\texttt{success} or \texttt{fail}}
    $i \gets 0;$ \\  
    $\ell \gets \floor{\log \lp( \frac{n}{2\log^{2}n}\rp)}$, $K \gets 2^{\ell};$ \\  
    $\mathcal{T} \gets \texttt{DomainPartitioning}\lp(\mathcal{D}, \mathcal{S}, \ell \rp);$   \\
    \While{\textnormal{\texttt{$A$ is not full}}}{
        $i \gets i+1$; $K_i \gets \frac{K}{2^{i-1}};$ \\  
        \texttt{array $A_i \gets$ allocate $\frac{n}{2^i}$ cells next to $A_{i-1}$ } \Comment*[r]{$A_0 \rightarrow$ array start}
        $C_i\gets \frac{A_{i}+N_{i-1}}{K_i}$ \Comment*[r]{set the capacity of each bin}
        $\texttt{for each }j\in [K_i],\ c^{(i)}_j = C_i-\texttt{empty}\lp(a^{(i-1)}_{2j-1}\rp)-\texttt{empty}\lp(a^{(i-1)}_{2j}\rp)$ \label{line:capacities}\\
        \If{$\textnormal{\texttt{there exists $j$, s.t.}}\ c^{(i)}_j \leq 0$}{\textbf{return} \texttt{fail}} \label{line:fail1a}
        \texttt{divide $A_i$ into $K_i$ consecutive buckets $a^{(i)}_j$ of sizes $c^{(i)}_j$, $j\in[K_i]$;} \\
        \If{\textnormal{$\texttt{remaining array places} \leq 100\log^2n$}}{$K_i = 1$, $A_i = a^{(i)}_1 = \texttt{remaining array places}$ \Comment*[r]{the last phase}}
        \While{\textnormal{\texttt{there does not exist$\ j$ s.t.$\ a^{(i)}_j$ \text{is full}}}}{
            $x\gets \texttt{receive\_sample}(\mathcal{D})$; \label{line:while}\\
            $k_1 \gets \texttt{index}(\mathcal{T}, i-1, x)$, $k_2 \gets \texttt{index}(\mathcal{T}, i,x)$; \\
            \eIf{\textnormal{\texttt{$a^{(i-1)}_{k_1}$ is not full}}}
            {$\texttt{InBucketPlacement}\lp(x, a^{(i-1)}_{k_1}\rp)$;}
            {$\texttt{InBucketPlacement}\lp(x, a^{(i)}_{k_2}\rp)$;}
            
        }
        \If{\textnormal{\texttt{$A_{i-1}$ is not full}}}{\textbf{return} \texttt{fail}} \label{line:fail2a}
        
    }
    \textbf{return} \texttt{success}
\end{algorithm}

\subparagraph*{Initialization.} First, we define $\ell$ as the unique positive integer such that $\frac{n}{4\log^{2}n}< 2^\ell \leq \frac{n}{2\log^{2}n}$. The initial number of buckets for the first subarray and balls-into-bins instance is $K = 2^\ell$. Next, we partition the domain $\mathcal{S}$ using the subroutine \texttt{DomainPartitioning}, which provides the partitioning of the domain into subsets (intervals / blocks) based on which we determine the bucket each element $x \in \mathcal{S}$ is inserted to. After completing the initializations, we move on to the first phase of the algorithm.

\subparagraph*{Phase 1.} Let $A_1$ be the subarray containing the first $\floor{\frac{n}{2}}$ elements of the original array. In the first phase we will place elements inside this subarray using a hash-based logic. We divide the $A_1$ elements into $K_1 = 2^\ell$ contiguous buckets, of size $C_1 = \frac{A_1}{K_1} = \Theta (\log^2 n)$ elements each. We place an arriving element $x\sim \mathcal{D}$ in the $k$-th bucket of $A_1$ if and only if $x \in T^{(1)}_k$, and then use the subroutine \texttt{InBucketPlacement} to position it within the bucket.

By definition of $\mathcal{T}$, each sample has equal probability of being placed into each bucket. Hence, placement of elements into the buckets of $A_1$ at the first phase of our algorithm can be thought as a balls-into-bins instance where we throw balls into $K_1$ bins, uniformly at random. When an element arrives and its designated bucket is full, we say that an overflow has occurred, triggering a transition to phase $2$.

\subparagraph*{Phase 2.} In order to handle overflows from subarray $A_1$ we allocate a subarray $A_2$ of size $\floor{\frac{n}{4}}$ next to $A_1$. Subarray $A_2$ is divided into $K_2 = K_1/2 = 2^{\ell-1}$ buckets. Regarding bucket sizes, crucially, the size of each bucket of $A_2$ is not the same, as discussed in the technical overview. To gain some intuition for the bucket sizes set in \Cref{line:capacities}, we consider the following example.

\subparagraph*{\textit{Example.}} Suppose for simplicity's sake, that we are on the classic online sorting case and subarray $A_1$ has $100$ cells and is divided into $K = K_1 = 4$ buckets, with capacity $C_1 = 25$ each, such that $T^{(1)}_1 = (0, 0.25),\ T^{(1)}_2 = (0.25, 0.5),\ T^{(1)}_3 =  (0.5, 0.75),\ T^{(1)}_4 = (0.75, 1)$ respectively. Assume that bucket $a^{(1)}_1$ is overflowed and the number of elements in each bucket are $(25, 20, 10, 15)$ at the time exactly before the overflow. Then, the up to this point empty, subarray $A_2$ will be brought into action to take care of the overflows of subarray $A_1$. By design, $A_2$ has $50$ cells and $K_2 = K_1/2= 2$ buckets that handle elements $T^{(2)}_1 = (0, 0.5),\ T^{(2)}_2 = (0.5, 1)$ respectively. If we allocate $25$ cells to each bucket in $A_2$, then we expect the first one to overflow significantly faster than the second one since in $A_1$ there are less empty cells for elements lying in $(0,0.5)$ than for elements lying in $(0.5,1)$. Thus, in order to handle the imbalances in $A_1$ we will allocate more cells in the first bucket in $A_2$. The total number of empty cells in $A_1$ is $\texttt{empty}(A_1) = 0+5+15+10 = 30$ and since we add $50$ cells with $A_2$ we have in total $80$ cells---this is the size of the second balls-into-bins instance. Specifically, we can consider that the instance has $2$ bins that handle elements from $(0, 0.5)$ and $(0.5, 1)$ respectively. The first bin contains the first two buckets of $A_1$ and the first bucket of $A_2$. The second bin contains the last two buckets of $A_1$ and the second bucket of $A_2$. We would like the capacities of the bins to be close enough in order to have a balanced balls-into-bins instance, thus equal to $\frac{\texttt{empty}(A_1) + A_2}{K_2} = 40$. Thus, we allocate $c^{(2)}_1 = 40-\texttt{empty}(a^{(1)}_1)-\texttt{empty}(a^{(1)}_2) = 35$ cells for the first bucket in $A_2$ and $c^{(2)}_2 = 40-\texttt{empty}(a^{(1)}_3)-\texttt{empty}(a^{(1)}_4) = 15$ cells for the second one. As a result, we obtain a new balanced balls-into-bins instance: the bins have equal capacity, so the expected overflow occurs later than it would if $A_2$ were simply divided into two equal-sized buckets.

This illustrates how $A_2$’s bucket sizes compensate for imbalances in $A_1$. We now return to the formal description of phase 2. The capacity of each bin in the second balls-into-bins instance is thus:

$$C_2 = \frac{\text{size of instance}}{\text{number of bins}}=\frac{A_2 + \e{A_1}}{K_2} = \Theta(\log^2 n)$$
(see also \Cref{lem:caplim}) and each bucket has capacity as defined in \Cref{line:capacities} in order to create a balanced instance, \textit{i.e.} $c^{(2)}_j = C_2-\texttt{empty}\lp(a^{(1)}_{2j-1}\rp)-\texttt{empty}\lp(a^{(1)}_{2j}\rp)$. If we cannot allocate such bucket size we say that our algorithm \emph{failed}. 

Regarding element placement, each bucket $j$ in $A_2$ accepts elements from the subset $T^{(2)}_j$. As elements arrive online, we first attempt to place them in their designated bucket in $A_1$; if it is full, we place them in the corresponding bucket in $A_2$. When an element arrives and its designated bucket in $A_1$ and $A_2$ is full, we transition to phase $3$ and we say that a bucket in array $A_2$ overflowed. Crucially, if $A_1$ is not full at this point, we say that the algorithm has \emph{failed}. In \Cref{sec:fail}, we show that our algorithm does not fail with high probability. 

\subparagraph*{Phases 3 and Beyond.} The same logic extends to all subsequent phases, while the main invariant remains: each subarray fills before a bucket of the next subarray overflows, otherwise our algorithm fails. When the remaining places in the array are less than $100\log^2 n$ we transition to the last phase of our algorithm. Note that for $C_i$ in each phase it holds (its proof is in the Appendix): 
\begin{lemma}\label{lem:caplim}
    For each phase $i$, it holds that the capacity $C_i$ of the $i$-th balls-into-bins instance is $\Theta(\log^2 n)$.
\end{lemma}
\subparagraph*{Final Phase.} In the final phase, the last subarray, denoted by $B$, is a single bucket, and elements are inserted using \texttt{InBucketPlacement}. Due to the above-discussed invariant, when $B$ overflows (\textit{i.e.}, when it is completely filled, as it has only one bucket), all previous subarrays are full. Thus, we have placed all elements into the array and the algorithm returns \texttt{success}. 

\subparagraph*{Number of Phases.} \label{lemma:B} Let $R$ be the number of phases of our algorithm and let $B$ be the last
subarray $A_R$. For $R$ it holds that $B + \sum_{i=1}^{R-1} \floor{\frac{n}{2^i}} = n$ with $B \le 100 \log^2 n$. Hence, 
\begin{align*}
    \sum_{i=1}^{R-1} \frac{n}{2^i} \geq n-100\log^2n \Rightarrow n(1-2^{-R}) \geq n-100\log^2n \Rightarrow R \geq \log \prn{\frac{n}{100\log^2 n}}
\end{align*}
and since $n-\sum_{i=1}^{R-2} \floor{\frac{n}{2^i}} > 100\log^2n$ we get that $R\leq 1+\log \prn{\frac{n}{100\log^2 n}}$, thus $R = \Theta(\log \frac{n}{\log^2n})$. Also, the size of $A_{R-1}$ is $n/2^{R-1} = \Theta(\log^2 n)$ and for $B$ it holds that $B \geq n-n(1-2^{-R}) = \Theta(\log^2 n)$, thus $B$ and $A_{R-1}$ are asymptotically tight, preserving the invariant.

\subsection{Probability of failure}
\label{sec:fail}

Note that sometimes our algorithm might fail. In this section we upper bound the probability of failure of our algorithm. There are two cases where our algorithm could fail. The first one is at line 10 of \Cref{theAlg} and the other is at line 23. To show that the failure probability $\Prob{\text{fail}}$ is small, we first present the following lemma, which is central to the analysis of the balls-into-bins instances arising in our problem.

\begin{lemma} \label{lem:balls-into-bins} Suppose a balls-into-bins instance with $K$ total bins, each with capacity $C = \Theta(\log^{2} n)$, and let $M = K\cdot C \leq n$ be the total number of balls that all bins can collectively hold. If the probability that an element belongs to each bin is the same, then the first overflow of a bin happens after $M-\frac{M}{\Theta(\log^{1/2} n)}$ balls are thrown w.p. at least $1-K/n^{2}$. Moreover, after throwing $M+\frac{M}{\Theta(\log^{1/2} n)}$ balls, all bins are full w.p. at least $1-K/n^{2}$. 
\end{lemma}

The proof of the lemma is deferred to the Appendix. We proceed by defining three types of events such that, conditioning on them, our algorithm cannot fail. We will conclude our proof that $\Prob{\text{fail}}$ is small by showing that these events occur with high probability. 

\begin{definition} \label{def:events}
    Consider the creation of the $i$-th balls-into-bins instance during the algorithm's execution. Denote by $M_i \coloneqq A_i + N_{i-1}$ its size and by $C_i$ its bins' capacity. We define the following events:
    \begin{itemize}
        \item $\mathcal{O}_i$: $A_i$ overflows after $M_i - \frac{M_i}{\Theta(\log^{1/2} n)}$ balls have been thrown in this instance
        \item $\mathcal{F}_i$: $A_i$ is fully filled after $M_i  + \frac{M_i}{\Theta(\log^{1/2} n)}$ balls have been thrown in this instance
        \item $\mathcal{G}_i$: each bin of the instance has received at least $C_i-\frac{C_i}{\Theta(\log^{1/3} n)}$ balls after $M_i - \frac{M_i}{\Theta(\log^{1/2} n)}$ balls have been thrown in this instance
    \end{itemize}
    and let $\mathcal{E}_i = \mathcal{O}_i\cap\mathcal{F}_i\cap \mathcal{G}_i$ and $\mathcal{E} = \cap_i \mathcal{E}_i$. 
\end{definition}

We now handle the first case of failure by proving the following lemma (the proof is in the Appendix). 

\begin{lemma} \label{lem:fail1}
    Conditioned on $\mathcal{E}$, it holds that $\forall j$, $C_i > \textnormal{\texttt{empty}}\lp(a^{(i-1)}_{2j-1}\rp)+\textnormal{\texttt{empty}}\lp(a^{(i-1)}_{2j}\rp),$
    for each phase $i$ of our algorithm.
\end{lemma}

We now proceed by handling the second case of failure. First we introduce some definitions. 

\begin{remark}
    Note that for phase $i$, the size of its balls-into-bins instance is $M_i = A_i+N_{i-1}$.
\end{remark}

\begin{definition}
    For each phase $i$, define as $T_i$ the number of elements inserted during phase $i$ until the first overflow in $A_i$ and $T_i'$ as the number of elements inserted since the beginning of phase $i$ until $A_i$ becomes full.
\end{definition}

We continue by proving the following lemma. 

\begin{lemma}
    Conditioned on $\mathcal{E}$, for each phase $i$, when the $i$-th balls-into-bins instance overflows, subarray $A_{i-1}$ is full. 
\end{lemma}

\begin{proof}
    We proceed using strong induction (see also \Cref{fig:timeline} for a visual representation of the proof).

\textbf{Base case, $i = 1$:} We aim to show that the first subarray $A_1$ is filled before any overflow occurs in the second balls-into-bins instance, that is, $T_2\geq T_1'-T_1$, conditioned on $\mathcal{E}$. For the first phase, it holds $T_1 = A_1- N_1$ (there is no previous phase). Thus, the size of the second balls-into-bins instance is $A_2+N_1$. Now, using \Cref{def:events}, observe that:
\begin{align*}
    T_2 - (T_1'-T_1) &\geq M_2 - \frac{M_2}{\Theta(\log^{1/2} n)} - \prn{M_1+\frac{M_1}{\Theta(\log^{1/2} n)} - A_1 + N_1} \\
    &\geq A_2+N_1 - \frac{A_2+N_1}{\Theta(\log^{1/2} n)} - \prn{A_1+\frac{A_1}{\Theta(\log^{1/2} n)} - A_1 + N_1} \\
    &= A_2 - \frac{A_2+N_1}{\Theta(\log^{1/2} n)} - \frac{A_1}{\Theta(\log^{1/2} n)} \\ 
    &> A_1/4 - o(A_1) \geq 0,
\end{align*}
since $A_2 > A_1/4$, thus conditioned on $\mathcal{E}$ we obtain the desired result, proving the base case. 

\textbf{Induction step:} Assume the statement holds for each $i\in [r-1]$. We aim to prove that it also holds for $i = r$. From induction hypothesis, since we condition on $\mathcal{E}$, every subarray $ A_j $ for $ j \in [r-1] $ becomes completely filled before any bucket in $ A_{j+1} $ overflows. As a result, when phase $r+1$ starts, every subarray $ A_j $ for $ j \in [r-1]$ is filled. This is crucial, since it implies that every new sample that arrives is inserted in this balls-into-bins instance; thus we can use \Cref{def:events} (otherwise the elements inserted into the instance are not necessarily uniform in bins). The size of the balls-into-bins instance of phase $r+1$ is $A_{r+1}+N_r$. Now, using \Cref{def:events} again, observe that:
\begin{align}
    T_{r+1}-(T_{r}'-T_{r}) &\geq  M_{r+1} - \frac{M_{r+1}}{\Theta(\log^{1/2} n)}-\prn{M_r+ \frac{M_r}{\Theta(\log^{1/2} n)}-M_r+ \frac{M_r}{\Theta(\log^{1/2} n)}} \\ 
    &=  A_{r+1}+N_r - \frac{A_{r+1}+N_r}{\Theta(\log^{1/2} n)}- 2\cdot\frac{A_r+N_{r-1}}{\Theta(\log^{1/2} n)} \\ 
    &\geq  A_{r+1}- 6\cdot\frac{A_{r-1}}{\Theta(\log^{1/2} n)} \qquad \text{(since $A_{r+1} < A_{r} < A_{r-1}$ and $A_i > N_i$)} \\ 
    &\geq A_{r-1}/8-o(A_{r-1}) \geq 0,
\end{align}

where we have used that $A_i = \Theta(A_{i-1})$, completing the proof of the induction step and thus proving the lemma.
\end{proof}

\begin{figure}[t]
    \centering
    \includegraphics[width=1\linewidth]{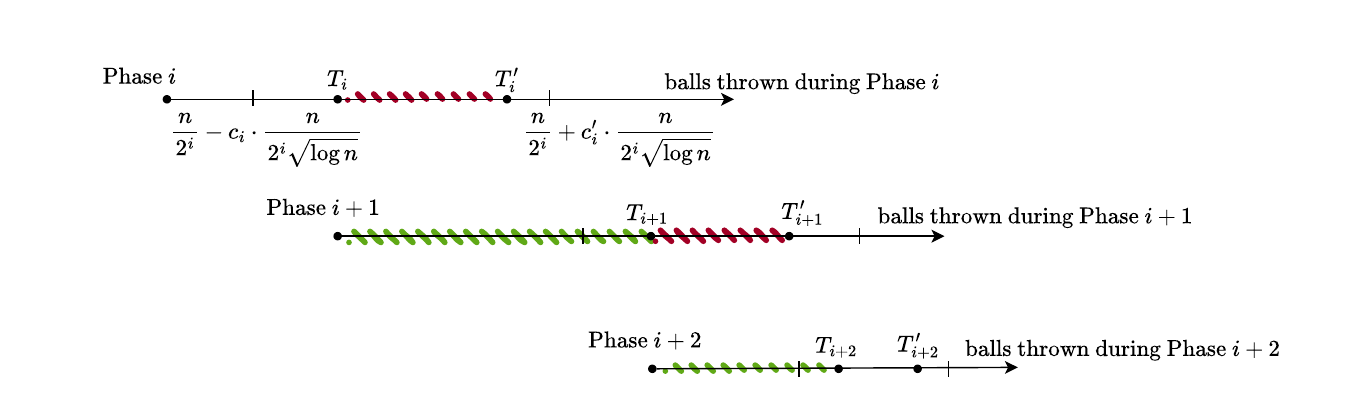}
    \caption{A timeline of the phases of our algorithm during execution. Note how the inequality $T_i'-T_i\leq T_{i+1}$ is preserved.} 
    \label{fig:timeline}
\end{figure}

We establish the following lemma, whose proof is deferred to the appendix.

\begin{lemma} \label{lem:E}
    It holds that $\Prob{\neg \mathcal{E}} \leq \frac{1}{n}.$
\end{lemma}

Since $\mathcal{E}$ implies that our algorithm will not fail we get that our algorithm does not fail with high probability since: 
    
$$\Prob{\text{fail}} \leq \Prob{\neg \mathcal{E}}<\frac{1}{n}$$

\begin{remark}
    For the full \Cref{theAlgFull} where we use $\widetilde{C}_i$ some buckets can have capacity $\floor{C_i}$ and others $\floor{C_i}+1$, thus $\widetilde{C}_i \geq \floor{C_i}$. For the time of the first overflow, from \Cref{lem:balls-into-bins} we get that w.p. $1-\frac{K_i}{n^2}$: $T_i \geq K_i\floor{C_i}-\frac{K_i\floor{C_i}}{\Theta(\log^{1/2} n)} \geq K_iC_i-K_i-\frac{K_iC_i}{\Theta(\log^{1/2} n)} \geq M_i - \frac{M_i}{\Theta(\log^{1/2} n)}$. 

    Similarly, for the other event, the capacity of each bin in the $i$-th balls-into-bins instance is at most $1+\floor{C_i}$. Thus for the time where all buckets are full, from \Cref{lem:balls-into-bins} we get that w.p. $1-\frac{K_i}{n^2}$: $T_i' \leq K_i(1+\floor{C_i})+\frac{K_i(1+\floor{C_i})}{\Theta(\log^{1/2} n)}\leq M_i + \frac{M_i}{\Theta(\log^{1/2} n)}$. 
\end{remark}


\section{Stochastic Online Sorting} 

In this section, we specialize \cref{theAlg} for the classical stochastic online sorting case where $\mathcal{D} = \mathcal{U}(0,1)$ and $\mathcal{S} = [0,1]$. We first, define subroutines \texttt{DomainPartitioning} and $\texttt{InBucketPlacement}$. 

\begin{itemize}
    \item \underline{\texttt{DomainPartitioning}$(\mathcal{U}, [0,1], \ell)$:} For $j \in [2^\ell]$, let: 
    
    $$T^{(1)}_j = [(j-1)/2^\ell, j/2^\ell) \text{ and } T^{(i)}_j = T^{(i-1)}_{2j-1} \cup T^{(i-1)}_{2j}$$

    \item \underline{\texttt{InBucketPlacement}$(x, a)$:} We use the deterministic adversarial algorithm $\mathcal{A}_{adv}$ of \cite{abrSODA} to place elements within each bucket.
\end{itemize}

We are now ready to prove the first main theorem of this work. 

\sorting*

\begin{proof}
We begin by conditioning on the success of our algorithm, which ensures that each bucket in every subarray stores exactly the elements belonging to its designated interval. We distinguish four sources of cost: the \emph{intra-bucket} cost, the \emph{inter-bucket} cost, the \emph{cost of connecting different subarrays} and the cost of the final \emph{subarray $B$}. 
\begin{itemize}
    \item \textbf{Intra-bucket:} Consider subarray $A_i$. Since we use $\mathcal{A}_{adv}$ algorithm for placement in each bucket, we incur total cost for the subarray: $\sum_{j=1}^{K_i}\sqrt{C_i}\cdot \frac{1}{K_i} = \Theta(\log n)$ and since we have $R-1 = O(\log n)$ such subarrays in total we incur total cost $O(\log^{2}n)$.
    
    \item \textbf{Inter-bucket:} Consider subarray $A_i$. Since, by definition, each bucket contains elements strictly smaller than the next, the total cost to "connect them" is at most $\sum_{j=1}^{K_i-1} \left( \frac{j+1}{K_i}-\frac{j-1}{K_i} \right) \leq 2$, thus in total for all subarrays: $O(\log n)$.
    
    \item \textbf{Between subarrays:} Between subarrays we incur cost at most $1$ and since we have $R = O(\log n)$ transitions this cost is $O(\log n)$.
    
    \item \textbf{Subarray $B$:} The elements are inserted in $B$ using $\mathcal{A}_{adv}$ thus we incur total cost $\Theta(\log n)$, since $B$ has $\Theta(\log^{2}n)$ elements. 
\end{itemize}

As a result, by conditioning on the success of our algorithm, we obtain total cost $O(\log^{2}n)$, thus the total cost is $O(\log^{2} n)$ with probability at least $1-1/n$, concluding the proof of the theorem.
 
\end{proof}

\begin{remark}
     Our algorithm is conceptually simple to state and analyze, and benefits from a hierarchical decomposition of Stochastic Online Sorting into balls-into-bins instances. The algorithm and the decomposition above can be naturally extended to non-uniform distributions. The only additional step we need to take care of is to partition the $[0,1]$ interval into as many intervals as required in each phase so that there is equal probability that a new point arrives in each interval. Furthermore, by a more careful analysis and parameter selection, we can show an upper bound of $O(\log^{3/2+\epsilon} n)$ on the total cost of the algorithm \textit{w.h.p.} for Stochastic Online Sorting. Finally, by choosing a sufficiently large constant~$c$ in the proof of \cref{lem:balls-into-bins}, we obtain an even stronger high-probability guarantee, e.g. $1 - 1/n^{100}$. 
\end{remark}


\section{Stochastic Online TSP}

We now extend our approach for stochastic online sorting to the \emph{Stochastic Online TSP} problem, obtained by increasing the dimensionality of the input. We show that \cref{theAlg} can be adapted to preserve the one-dimensional guarantees in arbitrary dimension $d$.

\tsp*

\subsection{\texttt{DomainPartitioning}}

\begin{figure}[t]
    \centering
    \includegraphics[width=0.6\linewidth]{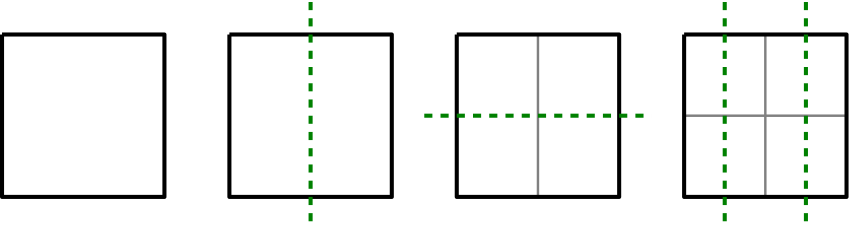}
    \caption{A visual representation of our splitting procedure on the plane. Uniform distribution ensures that splitting a block in half creates two blocks of equal probability mass.}
    \label{fig:blockformation}
\end{figure}

We partition the domain into hyperboxes, which we call \emph{blocks}. Let $\ell$ be the unique positive integer such that $\frac{n}{4\log^{2}n} < 2^\ell \leq \frac{n}{2\log^{2}n}.$ As in the one-dimensional case, the final block covers the entire domain, \textit{i.e.}, $T^{(\ell)}_1 = \mathcal{S}$. To construct the partition of $T^{(i)}$, we apply the split procedure from \Cref{alg:split} to $T^{(i+1)}$. Each split doubles the number of blocks, so after $\ell$ splits, $T^{(1)}$ consists of $2^\ell$ blocks. The splits proceed cyclically over the dimensions: the first along the first dimension, the second along the second, and so forth. After $d$ splits, we return to the first dimension and continue this process until all $\ell$ splits have been performed.

\begin{algorithm}
    \caption{\textsc{Split}$(\mathcal{B}, i)$}\label{alg:split}
    \KwData{Current set of blocks $\mathcal{B}$, splitting dimension $i$}
    \KwResult{Updated set of blocks after splitting}
    $\mathcal{B}' \gets \emptyset$\;
    \For{$b \in \mathcal{B}$}{
        Split $b$ into two equal halves $b_1, b_2$ along dimension $i$\;
        $\mathcal{B}' \gets \mathcal{B}' \cup \{b_1, b_2\}$\;
    }
    \Return{$\mathcal{B}'$}
\end{algorithm}

\subparagraph*{Elementary Blocks.} 
We refer to the blocks $T^{(1)} = \{b_1, b_2, \ldots, b_{2^\ell}\}$ at the lowest level of our hierarchy as \emph{elementary blocks}. Within an elementary block, all points are treated identically by \Cref{theAlg}: by construction, the algorithm’s decisions do not depend on which specific point from the block is presented. Consequently, we may assume that the points drawn within an elementary block $b$ are \textit{i.i.d.}\ samples from the uniform distribution $\mathcal{U}(b)$\footnote{Note that $\mathcal{U}(b)$ is the restriction of the global distribution $\mathcal{D}$ to block $b$.}.

\subparagraph{Block order.} At this stage, we impose an order on the blocks. Consecutive blocks are assigned to consecutive buckets/sub-subarrays. In addition, they are constructed so that they merge early during the execution of our algorithm. This motivates a further property: consecutive blocks must be neighbours, \textit{i.e.}, they share a face (or, more generally, a hyperface). We present \Cref{alg:findorder}, which traverses the blocks in such an order, and establish its correctness in \Cref{lem:algcor} which is formally proven in the Appendix.

\begin{algorithm}
    \caption{\textsc{Order}$(n_1,\dots,n_d)$}\label{alg:findorder}
    \KwData{Number of blocks $n_i$ along each dimension $i \in [d]$}
    \KwResult{Traversal order of all blocks}
    $v \gets (1,1,\dots,1)$ \Comment*[r]{coordinates of the first block}
    $L \gets [v]$ \Comment*[r]{list of visited block coordinates}
    $m \gets (1,1,\dots,1)$ \Comment*[r]{move direction in each dimension}

    \While{$\neg(v[1]=1 \wedge \dots \wedge v[d-1]=1 \wedge v[d]=n_d)$}{
        $j \gets 1$\;
        \While{$v[j] + m[j] \notin [1,n_j]$}{
            $m[j] \gets -m[j]$ \Comment*[r]{reverse direction in dim.\ $j$}
            $j \gets j+1$\;
        }
        $v[j] \gets v[j] + m[j]$, Append $v$ to $L$\;
    }
    \Return{$L$}\;
\end{algorithm}

\begin{lemma}\label{lem:algcor}
    \Cref{alg:findorder} visits every block exactly once, and any two consecutive blocks in the order are adjacent.
\end{lemma}

\subsection{\texttt{InBucketPlacement}}

Once mapped to a bucket, an element is placed in an empty cell using Bertram’s adversarial algorithm \cite{bertESA}.


\subsection{Cost Analysis}

Similarly to the stochastic online sorting case, we distinguish three sources of cost: the \emph{intra-bucket} cost, the \emph{inter-bucket} cost, and the \emph{cost of connecting different subarrays}. Since absolute values are not informative, all bounds will be expressed in terms of $\OPT$. Before estimating $\OPT$ and analyzing each cost source separately, we recall several key results that will be used throughout the analysis.

\begin{theorem}[\cite{Beardwood_Halton_Hammersley_1959}]
    Given sufficiently large $n$, the expected length of a TSP tour with $n$ points drawn i.i.d. from $\mathcal{U}([0,\Delta]^d)$ is $\beta_d\cdot n^{1-1/d}\cdot \Delta$, where $\beta_d\approx\sqrt{\frac{d}{2\pi e}}$. We denote this quantity by $\text{TSP}(n,d,\Delta)$. 
\end{theorem} 

\begin{corollary}\label{exOPT}
    Since $\OPT$ represents the length of the optimal tour of $n$ points drawn i.i.d. from $\mathcal{U}([0,1]^d)$ it is $\E{\OPT} = \text{TSP}(n,d,1) = \beta_d \cdot n^{1 - 1/d},$ where the randomness is over the drawn instance.
\end{corollary}

\begin{proposition}
The following properties hold:
\begin{enumerate}
    \item For all $m<n$ and all $d,\Delta$, we have $\text{TSP}(m,d,\Delta) < \text{TSP}(n,d,\Delta)$.
    
    \item For any convex $S \subseteq [0,\Delta]^d$, let $tour(S,n)$ denote the expected length of the optimal TSP tour on $n$ points drawn uniformly i.i.d.\ in $S$. Then, $tour(S,n) \;\leq\; \text{TSP}(n,d, \max_{x,y \in S}  \|x- y\|_{\infty}) \;\leq\; \text{TSP}(n,d,\Delta).$
\end{enumerate}
\end{proposition}

\begin{lemma}[Chapter 2, \cite{steele}]\label{conc}
Let $T$ be the length of the TSP of $n$ points drawn i.i.d. in $[0, \Delta]^d$. It holds for some constant $c$:

$$\Prob{\abs{T-\text{TSP}(n,d,\Delta)}\geq t}\leq 2\exp \left(-\frac{ct^2}{n^{1-2/d}\Delta^2}\right)$$
\end{lemma}

\subparagraph*{Estimating the Optimal Cost.}

In contrast to the one-dimensional case, there is no simple closed-form expression for the optimal cost in higher dimensions. We therefore rely on asymptotic estimates and concentration bounds to control the value of $\OPT$.

\begin{lemma}\label{lem:OPT}
With probability at least $1 - 2\exp(-c' d n)$, the optimal TSP tour satisfies

$$
    \OPT \;\ge\; \tfrac{1}{2}\,\beta_d\, n^{1-1/d}.
$$
\end{lemma}

\begin{figure}[t]
    \centering
    \includegraphics[width=0.4\linewidth]{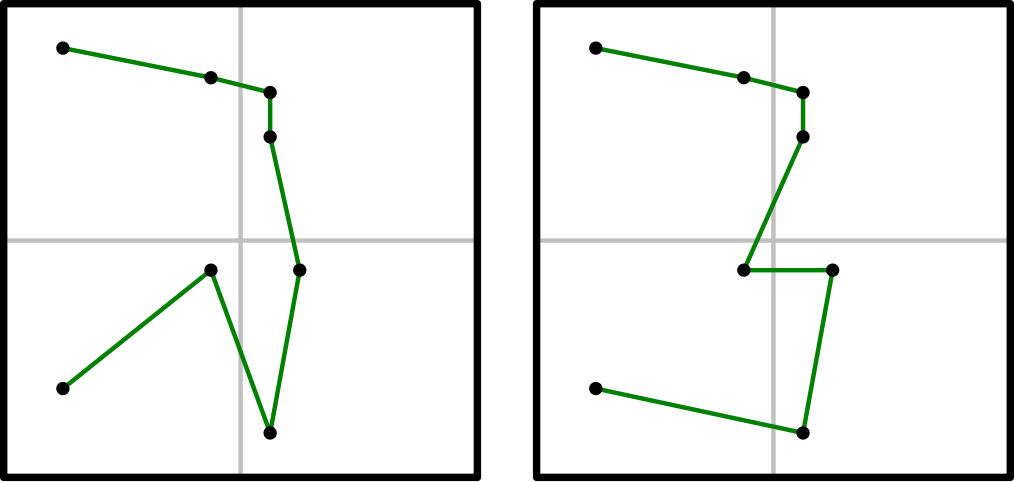}
    \caption{The domain $[0,1] \times [0,1]$ is partitioned into four blocks. 
The left panel shows the block-by-block tour produced by our algorithm, 
while the right panel shows the optimal TSP tour on the same instance. 
Although the two tours may differ, our partitioning ensures that their lengths 
remain close with high probability.}

    \label{fig:optpath}
\end{figure}

We now proceed to the analysis of the three sources of cost. 

\subparagraph*{Intra-Bucket Cost.}

\begin{lemma}\label{oneblock}
Let $\OPT_i$ be the length of the optimal tour of $6 \log^2 n$ points drawn uniformly at random from elementary block $b_i$. Then, 

$$\Prob{\OPT_i \,\geq\, 72 \beta_d\, n^{-1/d}\, \log^2 n} \;\leq\; \frac{2}{n^2}.$$
\end{lemma}

\begin{proof}
\textbf{Step 1: Block side length.}  
Since the domain is split $\ell$ times, each dimension is split at least $\lfloor \ell/d \rfloor$ times\footnote{We assume that $\ell \geq d$. In the \Cref{ellgeqd} we remove this assumption.}, thus the maximum side length of a block is at most $(1/2)^{\lfloor \ell/d \rfloor}$. As $\ell$ and $d$ are integers, $(1/2)^{\lfloor \ell/d \rfloor}   \;\leq\; (1/2)^{(\ell-d)/d} \;<\; \Bigl(\tfrac{4\log^2 n}{n}\Bigr)^{1/d}\cdot 2 \;\leq\; 4\Bigl(\tfrac{\log^2 n}{n}\Bigr)^{1/d}.$ Hence, each side of $b_i$ has length at most $4(\tfrac{\log^2 n}{n})^{1/d}$.

\noindent\textbf{Step 2: Expected cost.}  The expected length of the optimal subtour inside $b_i$ satisfies:

$$
    \E{\OPT_i} \;\leq\; \text{TSP}(6\log^2 n,d,\Delta)
    \;\leq\; \beta_d (6\log^2 n)^{1-1/d}\, \Delta,
$$
where $\Delta$ is the side length of $b_i$. Substituting $\Delta \leq 4(\tfrac{\log^2 n}{n})^{1/d}$ gives:

$$
    \E{\OPT_i} \;\leq\; 24\beta_d\, n^{-1/d}\, \log^2 n.
$$
\textbf{Step 3: Concentration.}  
Applying \Cref{conc} with $t = 48\beta_d\, n^{-1/d}\, \log^2 n$, we bound the probability that $\OPT_i$ exceeds three times its expectation:

$$
    \Prob{\OPT_i \geq 72\beta_d\, n^{-1/d}\, \log^2 n}
    \;\leq\; 2\exp\!\left(- \frac{c'' d\, n^{-2/d}\, (\log^2 n)^2 \cdot n^{2/d}}{\log^2 n}\right) \;\leq\; \frac{2}{n^2}.
$$
\end{proof}

\begin{corollary}\label{cor:blockcost}
With probability at least $1 - \tfrac{1}{n\log^2 n}$, every one of the $2^\ell \leq \tfrac{n}{2\log^2 n}$ elementary block tours has length at most $72\beta_d\, n^{-1/d}\, \log^2 n.$
\end{corollary}

So far, we proved that the length of the optimal tour of $6 \log^2 n$ points uniformly chosen in a block is bounded by $72\beta_d\, n^{-1/d}\, \log^2 n$ with probability at least $1 - 2/n^2$. We want to transition from this to a bound of the optimal total intra bucket cost of a subarray $j$, denoted by $\textsc{InBucket}(j)$. For the first subarray, elementary blocks are mapped one to one to buckets and every bucket contains strictly less than $6 \log^2n$ points therefore \Cref{cor:blockcost} implies that $\textsc{InBucket}(1) \le 72\beta_d\, n^{1-1/d}$ with probability at least $1 - \tfrac{1}{n\log^2 n}$. We now show that any subsequent (fixed) subarray enjoys a comparable bound.

\begin{lemma}\label{lem:merge-dominance}
Fix an arbitrary phase $j$ and an arbitrary block $B_k^{(j)}$ of phase $j$. Block $B_k^{(j)}$ consists of some elementary blocks $b_{k_1}, \ldots, b_{k_m}$. Fix any realization of points and let $r(B)$ denoted the realized points inside block $B$,

$$
\textsc{InBucket}(j) = \sum_{i=1}^{{2^{\ell - j + 1}}} \OPT(r(B^{(j)}_i)) \;\le\; \sum_{i=1}^{2^\ell} \OPT(r(b_i)) \;+\; \beta_d\, n^{1-1/d}.
$$
where $\OPT(S)$ denotes the length of the optimal tour of the points in set $S$.
\end{lemma}

\begin{proof}
A feasible tour inside $r(B^{(j)}_k)$ is obtained by concatenating the optimal tours of $r(b_{k_1}),\dots,r(b_{k_m})$ and adding at most $(r-1)$ connecting edges between adjacent elementary blocks. Summed over all phase $j$ blocks, the total number of such connectors is at most the number of elementary blocks. By \Cref{lem:connect-blocks}, each connector has length at most $2 \beta_d \; n^{-1/d}\log^2 n$, and by \Cref{cor:subarray} the total intra-subarray connection cost is deterministically bounded by $\beta_d\, n^{1-1/d}$. This yields the claimed inequality.
\end{proof}

At this point we restate \cref{lem:caplim} slightly changing the phrasing to fit current context. 

\begin{lemma}
    For every subarray $j$, each bucket has capacity less than $6 \log^2 n$ points.
\end{lemma}

\begin{corollary}\label{pointlim}
    Each subarray gets at most $6 \log^2 n$ points from any elementary block.
\end{corollary}

\begin{lemma}\label{dominion}
    For any $t \in \mathbb{R}$ and any subarray $j$, it holds $\Prob{\OPT_i \geq t} \geq \Prob{\OPT(r(b_i)) \geq t}$.
\end{lemma}

\begin{proof}
The cardinality of $r(b_i)$ is a random variable. However, by \Cref{pointlim}, $r(b_i)$ contains fewer than $6 \log^2(n)$ points. In both cases, the points are independently drawn uniformly at random within the same block. The proof then follows from the observation that the probability of the optimal tour length of $x$ i.i.d.\ points in some space $S$ exceeding a threshold $t$ is non-decreasing in the number of points $x$.
\end{proof}

\begin{proposition}\label{prop:phase-symmetry}
Fix any subarray $j$. With probability at least $1 - \tfrac{1}{n\log^2 n}$, it is $\textsc{InBucket}(j) \leq 73\beta_d\, n^{1-1/d}$.
\end{proposition}

\begin{proof}
\cref{dominion} combined with \cref{oneblock} implies $\Prob{\OPT(r(b_i)) \geq 72\beta_dn^{-1/d}\log^2n} \le \frac{2}{n^2}$. Taking a union bound over different elementary blocks and adding the connecting term completes the proof.
\end{proof}

\begin{corollary}
    By taking a union bound over the subarrays we get that, with probability at least $1 - \frac{1}{n \log n}$, for every subarray $j$ it holds $\textsc{InBucket}(j) \leq 146 \OPT$.
\end{corollary}

\subparagraph*{Inter-Bucket Cost.}

\begin{lemma}\label{lem:connect-blocks}
The cost of connecting two neighbouring blocks is at most 

$$
    2\sqrt{d}\cdot 4\Bigl(\tfrac{\log^{2}n}{n}\Bigr)^{1/d}
    \;\leq\; 2\beta_d\, n^{-1/d}\, \log^{2}n.
$$
\end{lemma}

\begin{remark}
Lemma~\ref{lem:connect-blocks} applies to elementary blocks. When moving to higher phases, blocks are formed by merging pairs of blocks from the previous phase. As a result, the diameter of a block (and hence the cost of connecting two neighboring blocks) increases by at most a factor strictly smaller than $2$ from one phase to the next. However, at the same time, the number of blocks to be connected decreases by a factor of $2$ at each phase.  Therefore, the total cost incurred by connecting all blocks within a subarray remains of the same order, which leads to Corollary~\ref{cor:subarray}.
\end{remark}

\begin{corollary}\label{cor:subarray}
Using the observation above, the cost of connecting all blocks within a subarray is deterministically bounded by $\beta_d \cdot n^{1 - 1/d}$.
\end{corollary}

\subparagraph*{Cost of Connecting Subarrays.}

\begin{lemma}\label{lem:between-subarrays}
The cost of connecting two subarrays is at most $\OPT$.
\end{lemma}

\begin{proof}
Connecting subarrays requires adding a single edge between two points, one from each subarray. Since the optimal tour $\OPT$ already spans all points in the domain, the additional cost of this connection is trivially bounded above by $\OPT$.
\end{proof}

\subparagraph{Putting Everything Together.}

We are now ready to prove \Cref{thm:2}. We decompose the total cost into four components:

\begin{itemize}
    \item \textbf{Intra-bucket:} 
    For each subarray $A_j$, placements inside buckets are handled by Bertram’s algorithm \cite{bertESA}. This incurs cost $\sum_{i=1}^{K_j} \sqrt{C_j}\cdot \OPT(r(B_i)) \leq \Theta(\log n) \sum_{i=1}^{2^{\ell}} \OPT_i = \Theta(\log n)\OPT$. where the first inequality is derived from \Cref{lem:merge-dominance} and the second from \Cref{oneblock}. Since there are $R = O(\log n)$ subarrays, the total bucket cost is $O(\log^2 n)\OPT$.

    \item \textbf{Inter-bucket:} 
    By \Cref{lem:connect-blocks,cor:subarray}, the additional cost of connecting all buckets within a subarray $A_i$ is at most $O(\OPT)$. Summing over all $k = O(\log n)$ subarrays yields a total of $O(\log n)\OPT$.

    \item \textbf{Between subarrays:}  
    Connecting consecutive subarrays requires at most one additional edge per transition. By \Cref{lem:between-subarrays}, this incurs cost at most $\OPT$ per transition. With $k = O(\log n)$ subarrays, the total cost is $O(\log n)\OPT$.

    \item \textbf{Subarray $B$:}  
    The remaining elements are placed in $B$ using Bertram’s algorithm \cite{bertESA}, incurring total cost $O(\log n)\OPT$, since $|B| = O(\log^2 n)$.
\end{itemize}

Conditioning on the high-probability event that \Cref{cor:blockcost} holds for all subarrays $A_i$ and conditioning on the event that the algorithm does not fail, the total cost of the algorithm is $O(\log^2 n)\OPT$. This completes the proof of \Cref{thm:2}.

\section{Conclusion and Open Directions}

In this work, we study the stochastic variants of Online Sorting and Online TSP, obtaining polylogarithmic upper bounds in both settings. Analyzing the problem through the lens of the balls-into-bins paradigm reveals exponential improvements over previous approaches. A natural direction is to extend our bounds to general known distributions in higher dimensions. At the same time, our framework also assumes full knowledge of the distribution. Relaxing this assumption prompts several questions: What guarantees are possible when points are drawn \textit{i.i.d.} from an \emph{unknown} distribution, or when inputs arrive in uniformly random order with no distributional assumptions?



\bibliographystyle{alpha}
\bibliography{main.bib}

@inproceedings{abrSODA,
  author       = {Anders Aamand and
                  Mikkel Abrahamsen and
                  Lorenzo Beretta and
                  Linda Kleist},
  editor       = {Nikhil Bansal and
                  Viswanath Nagarajan},
  title        = {Online{ S}orting and {T}ranslational {P}acking of {C}onvex {P}olygons},
  booktitle    = {Proceedings of the 2023 {ACM-SIAM} Symposium on Discrete Algorithms,
                  {SODA} 2023, Florence, Italy, January 22-25, 2023},
  pages        = {1806--1833},
  publisher    = {{SIAM}},
  year         = {2023},
  url          = {https://doi.org/10.1137/1.9781611977554.ch69},
  doi          = {10.1137/1.9781611977554.CH69},
  bibsource    = {dblp computer science bibliography, https://dblp.org}
}

@inproceedings{AbrESA,
  author       = {Mikkel Abrahamsen and
                  Ioana O. Bercea and
                  Lorenzo Beretta and
                  Jonas Klausen and
                  L{\'{a}}szl{\'{o}} Kozma},
  editor       = {Timothy M. Chan and
                  Johannes Fischer and
                  John Iacono and
                  Grzegorz Herman},
  title        = {Online {S}orting and {O}nline {TSP:} {R}andomized, {S}tochastic, and {H}igh-{D}imensional},
  booktitle    = {32nd Annual European Symposium on Algorithms, {ESA} 2024, September
                  2-4, 2024, Royal Holloway, London, United Kingdom},
  series       = {LIPIcs},
  volume       = {308},
  pages        = {5:1--5:15},
  publisher    = {Schloss Dagstuhl - Leibniz-Zentrum f{\"{u}}r Informatik},
  year         = {2024},
  url          = {https://doi.org/10.4230/LIPIcs.ESA.2024.5},
  doi          = {10.4230/LIPICS.ESA.2024.5},
  bibsource    = {dblp computer science bibliography, https://dblp.org}
}

@article{space1,
author = {Platzman, Loren K. and Bartholdi, John J.},
title = {Spacefilling curves and the planar travelling salesman problem},
year = {1989},
issue_date = {Oct. 1989},
publisher = {Association for Computing Machinery},
address = {New York, NY, USA},
volume = {36},
number = {4},
issn = {0004-5411},
url = {https://doi.org/10.1145/76359.76361},
doi = {10.1145/76359.76361},
journal = {J. ACM},
month = oct,
pages = {719–737},
numpages = {19}
}

@misc{hu2025,
      title={Nearly Optimal Bounds for Stochastic Online Sorting}, 
      author={Yang Hu},
      year={2025},
      eprint={2508.07823},
      archivePrefix={arXiv},
      primaryClass={cs.DS},
      url={https://arxiv.org/abs/2508.07823}, 
}

@article{space2,
title = {Worst-case examples for the spacefilling curve heuristic for the Euclidean traveling salesman problem},
journal = {Operations Research Letters},
volume = {8},
number = {5},
pages = {241-244},
year = {1989},
issn = {0167-6377},
doi = {https://doi.org/10.1016/0167-6377(89)90047-3},
url = {https://www.sciencedirect.com/science/article/pii/0167637789900473},
author = {Dimitris Bertsimas and Michelangelo Grigni}
}

@inproceedings{space3,
author = {Hajiaghayi, Mohammad T. and Kleinberg, Robert and Leighton, Tom},
title = {Improved lower and upper bounds for universal TSP in planar metrics},
year = {2006},
isbn = {0898716055},
publisher = {Society for Industrial and Applied Mathematics},
address = {USA},
booktitle = {Proceedings of the Seventeenth Annual ACM-SIAM Symposium on Discrete Algorithm},
pages = {649–658},
numpages = {10},
location = {Miami, Florida},
series = {SODA '06}
}

@inproceedings{chrissgour,
  author       = {George Christodoulou and
                  Alkmini Sgouritsa},
  editor       = {Philip N. Klein},
  title        = {An Improved Upper Bound for the Universal {TSP} on the Grid},
  booktitle    = {Proceedings of the Twenty-Eighth Annual {ACM-SIAM} Symposium on Discrete
                  Algorithms, {SODA} 2017, Barcelona, Spain, Hotel Porta Fira, January
                  16-19},
  pages        = {1006--1021},
  publisher    = {{SIAM}},
  year         = {2017},
  url          = {https://doi.org/10.1137/1.9781611974782.64},
  doi          = {10.1137/1.9781611974782.64},
  bibsource    = {dblp computer science bibliography, https://dblp.org}
}

@misc{bertESA,
      title={Online {M}etric {TSP}}, 
      author={Christian Bertram},
      year={2025},
      eprint={2504.17716},
      archivePrefix={arXiv},
      primaryClass={cs.DS},
      url={https://arxiv.org/abs/2504.17716}, 
}

@article{Beardwood_Halton_Hammersley_1959, title={The shortest path through many points}, volume={55}, DOI={10.1017/S0305004100034095}, number={4}, journal={Mathematical Proceedings of the Cambridge Philosophical Society}, author={Beardwood, Jillian and Halton, J. H. and Hammersley, J. M.}, year={1959}, pages={299–327}}

@book{steele,
author = {Steele, J. Michael},
title = {Probability Theory and Combinatorial Optimization},
publisher = {Society for Industrial and Applied Mathematics},
year = {1997},
doi = {10.1137/1.9781611970029},
address = {},
edition   = {},
URL = {https://epubs.siam.org/doi/abs/10.1137/1.9781611970029},
eprint = {https://epubs.siam.org/doi/pdf/10.1137/1.9781611970029}
}

@article{anderssen,
 ISSN = {00361399},
 URL = {http://www.jstor.org/stable/2100577},
 author = {R. S. Anderssen and R. P. Brent and D. J. Daley and P. A. P. Moran},
 journal = {SIAM Journal on Applied Mathematics},
 number = {1},
 pages = {22--30},
 publisher = {Society for Industrial and Applied Mathematics},
 title = {Concerning integral(0 to 1) ... integral(0 to 1) (x21 + ... + x2k)1/2 dx1 ... dxk and a Taylor Series Method},
 urldate = {2025-07-14},
 volume = {30},
 year = {1976}
}

@misc{azar2025,
      title={Nearly Tight Bounds for the Online Sorting Problem}, 
      author={Yossi Azar and Debmalya Panigrahi and Or Vardi},
      year={2025},
      eprint={2508.14287},
      archivePrefix={arXiv},
      primaryClass={cs.DS},
      url={https://arxiv.org/abs/2508.14287}, 
}

@misc{nirjhor2025,
      title={Improved Online Sorting}, 
      author={Jubayer Nirjhor and Nicole Wein},
      year={2025},
      eprint={2508.14361},
      archivePrefix={arXiv},
      primaryClass={cs.DS},
      url={https://arxiv.org/abs/2508.14361}, 
}

@misc{kuszmaul2016,
      title={Fast Concurrent Cuckoo Kick-Out Eviction Schemes for High-Density Tables}, 
      author={William Kuszmaul},
      year={2016},
      eprint={1605.05236},
      archivePrefix={arXiv},
      primaryClass={cs.DS},
      url={https://arxiv.org/abs/1605.05236}, 
}

@inproceedings{fotakis2003,
author = {Fotakis, Dimitris and Pagh, Rasmus and Sanders, Peter and Spirakis, Paul},
year = {2003},
month = {02},
pages = {271-282},
title = {Space Efficient Hash Tables with Worst Case Constant Access Time},
isbn = {978-3-540-00623-7},
doi = {10.1007/3-540-36494-3_25}
}

@article{ausiello99,
author = {Ausiello, Giorgio and Feuerstein, Esteban and Leonardi, Stefano and Stougie, Leen and Talamo, Maurizio},
year = {1999},
month = {08},
pages = {},
title = {Algorithms for the On-line Travelling Salesman}
}

\appendix

\section{Full Pseudocode}

\begin{algorithm}[H]
    \caption{General Algorithmic Framework - Detailed}\label{theAlgFull} 
    \KwData{Array $A[1:n]$, Distribution $\mathcal{D}$ over a domain $\mathcal{S}$, subroutines \texttt{receive\_sample, DomainPartitioning, InBucketPlacement, index, empty}}
    \KwResult{\texttt{success} (placement of $n$ samples from $\mathcal{D}$ into $A$) or \texttt{fail}}
    $i \gets 0$ \Comment*[r]{the current phase of the algorithm}
    $\texttt{ptr} \gets 0$ \Comment*[r]{the index just before the next subarray to allocate}
    $\ell \gets \floor{\log \lp( \frac{n}{2\log^{2}n}\rp)}$, $K \gets 2^\ell$\Comment*[r]{$K$ is the initial number of buckets}  
    $\mathcal{T} \gets \texttt{DomainPartitioning}\lp(\mathcal{D}, \mathcal{S}, \ell \rp)$ \Comment*[r]{the disjoint subset binary tree}  
    \While{$A$ is not full}{
        $i \gets i+1$, $K_i = \frac{K}{2^{i-1}}$; \\
        $A_i \gets A\lp[\texttt{ptr}+1 : \texttt{ptr}+1 + \floor{\frac{n}{2^i}}\rp]$; \\
        $\texttt{idx}\gets 0$, $C_i\gets \floor{\frac{A_{i}+\texttt{empty}(A_{i-1})}{K_i}}$, $m = A_{i}+\texttt{empty}(A_{i-1}) \mod K_i$; \\
        \For{$j\in [K_i]$}{
            \If{$j \leq m$}{$\widetilde{C}_i = C_i+1$}
            \Else{$\widetilde{C}_i = C_i$}
            $c^{(i)}_j = \widetilde{C}_i-\texttt{empty}\lp(a^{(i-1)}_{2j-1}\rp)-\texttt{empty}\lp(a^{(i-1)}_{2j}\rp)$ \\
            \If{$c^{(i)}_j \leq 0$}{\textbf{return} \texttt{fail}} \label{line:fail1}
            $a^{(i)}_j = A_i[\texttt{idx}+1:\texttt{idx}+1+c^{(i)}_j]$; \\ 
            $\texttt{idx} \gets \texttt{idx} + c^{(i)}_j$
        }
        \If{$n-\texttt{ptr} \leq 100\log^2n$}{$K_i = 1$, $A_i = a^{(i)}_1 = A[\texttt{ptr+1}:n]$ \Comment*[r]{the last phase of $\ALG$ }}
        \While{$\forall j \in [K_i],\ a^{(i)}_j$ \text{is not full}}{
            $x\gets \texttt{receive\_sample}(\mathcal{D})$; \\
            $k_1 \gets \texttt{index}(\mathcal{T}, i-1, x)$, $k_2 \gets \texttt{index}(\mathcal{T}, i, x)$; \\
            \eIf{$a^{(i-1)}_{k_1}$ is not full}
            {$\texttt{InBucketPlacement}\lp(x, a^{(i-1)}_{k_1}\rp)$;}
            {$\texttt{InBucketPlacement}\lp(x, a^{(i)}_{k_2}\rp)$;}
            
        }
        \If{$A_{i-1}$ is not full}{\textbf{return} \texttt{fail}} \label{line:fail2}
        $\texttt{ptr} \gets \texttt{ptr} + A_i$; \\ 
    }
    \textbf{return} \texttt{success}
\end{algorithm}

\section{Deferred Proofs}

\subsection{Proof of Lemma~\ref{lem:caplim}}

\begin{proof}
For $C_i = \frac{\texttt{empty}(A_{i-1})+A_i}{K_i} = \frac{N_{i-1}+A_i}{K_i}$ it holds that 
$$\frac{A_i}{K_i} \leq \frac{N_{i-1}+A_i}{K_i} \leq \frac{A_{i-1}+A_i}{K_i} \leq \frac{3A_i}{K_i}$$ and since 

$$\frac{A_i}{K_i} = \frac{\floor{n/2^i}}{2^{\ell-i+1}} \in [0.5\log^2 n,  2\log^2 n],$$ 
we get that $C_i \in  [0.5\log^2 n,  6\log^2 n]$ thus proving the lemma.     
\end{proof}

\subsection{Proof of Lemma~\ref{lem:balls-into-bins}}

\begin{proof} 
We will use standard Chernoff bounds along with a union bound over bad events. For each $i \in [K]$, let $L_i^{(t)}$ denote the random variable representing the number of balls in the $i$-th bin after $t$ balls have been thrown. Note that $\mathbb{E}[L_i^{(t)}] = t/K$. Also, let $X_i^{(t)}\sim \text{Bernoulli}(1/K)$ be the indicator random variable that is equal to $1$ if and only if the $t$-th ball lands in bin $i$. Thus, $L_i^{(t)} = \sum_{j \in [t]} X_{i}^{(j)}$, where $X_{i}^{(j)}$ are independent. First, to bound the time of the first overflow, we apply a standard upper tail Chernoff bound. After throwing $T = M - c\cdot \frac{M}{\log^{1/2} n} = M -\frac{M}{\Theta(\log^{1/2} n)}$ balls for some positive constant $c$, we obtain that ($c$ is not absorbed into $\Theta(\cdot)$ term):

$$ \Prob{L_i^{(T)} \geq C} \leq \Exp{-\frac{\prn{\frac{M-T}{T}}^2\cdot \frac{T}{K}}{\frac{M-T}{T} + 2}} = \Exp{-\frac{(M-T)^2C}{(M+T)M}} \leq \Exp{-\frac{c^2}{2}\Theta(\log n)} \leq \frac{1}{n^2} $$
for some choice of constant $c$, depending on $C = \Theta(\log^{2} n)$. Taking a union bound over all $K$ bins, we get:

$$ \Prob{\exists i \in [K] : L_i^{(T)} \geq C} \leq \frac{K}{n^2}.$$

This shows that with probability at least $1 - \frac{K}{n^2}$, no bin has exhibited overflow after $T = M - o(M)$ balls have been thrown, completing the first part of the lemma.

For the second part of the lemma, we want to show that after throwing  $T' = M + c\cdot \frac{M}{\log^{1/2} n} = M + \frac{M}{\Theta(\log^{1/2} n)}$  balls for some positive constant $c$, all bins are full with probability at least $1-\frac{K}{n^2}$. From a standard lower tail Chernoff bound we obtain that (again $c$ is not absorbed into $\Theta(\cdot)$ term):

$$ \Prob{L_i^{(T')} < C} \leq \Exp{-\frac{\prn{\frac{T'-M}{T'}}^2\cdot \frac{T'}{K}}{2}} = \Exp{-\frac{(T'-M)^2C}{2T'M}} \leq \Exp{-\frac{c^2}{4}\Theta(\log n)} \leq \frac{1}{n^2} $$
for some choice of constant $c$, depending on $C = \Theta(\log^{2} n)$. Taking a union bound over all $K$ bins, we get:

$$\Prob{\exists i \in [K] : L_i^{(T')} < C} \leq \frac{K}{n^2}, $$
finishing the proof. 
\end{proof}

\subsection{Proof of Lemma~\ref{lem:fail1}}

\begin{proof}
    Fix a phase $i$. A new phase $i$ has just started, thus one bin of phase $i-1$ has overflowed. Since we condition on $\mathcal{O}_{i-1}$, $M_{i-1} - \frac{M_{i-1}}{\Theta(\log^{1/2} n)}$ balls have been thrown in this instance. Thus, due to the conditioning on $\mathcal{G}_{i-1}$, we also get that each bin of the instance has received at least $C_{i-1}-\frac{C_{i-1}}{\Theta(\log^{1/3} n)}$ balls. Since, the number of empty cells in each sub-subarray of $A_{i-1}$ is less than the remaining capacity of the bin we get that $\forall j,\ \texttt{empty}\lp(a^{(i-1)}_{j}\rp) < \frac{C_{i-1}}{\Theta(\log^{1/3} n)}$. Thus, we get that: 

    $$ \textnormal{\texttt{empty}}\lp(a^{(i-1)}_{2j-1}\rp)+\textnormal{\texttt{empty}}\lp(a^{(i-1)}_{2j}\rp) < 2\frac{C_{i-1}}{\Theta(\log^{1/3} n)} < C_i,$$
    since $C_i = \Theta(\log^2 n)$ for each phase $i$, finishing the proof. 
\end{proof}

\subsection{Proof of Lemma~\ref{lem:E}}

\begin{proof}
    For the events $\mathcal{O}_i, \mathcal{F}_i$,  we get immediately from \Cref{lem:balls-into-bins} that they occur with probability at least $1-{K_i}/{n^2}$. For the events $\mathcal{G}_i$, we use the second part of \Cref{lem:balls-into-bins} for $C = C_i-\frac{C_i}{\Theta(\log^{1/3} n)}$ and $K = K_i$. Specifically, since: 
    
    $$M + \frac{M}{\Theta(\log^{1/2} n)} = K_i\prn{C_i-\frac{C_i}{\Theta(\log^{1/3} n)}} + \frac{K_i\prn{C_i-\frac{C_i}{\Theta(\log^{1/3} n)}}}{\Theta(\log^{1/2}n)} < M_i-\frac{M_i}{\Theta(\log^{1/2} n)},$$
    after $M + \frac{M}{\Theta(\log^{1/2} n)} < M_i-\frac{M_i}{\Theta(\log^{1/2} n)}$ balls are thrown each bin is full, thus has at least $C = C_i-\frac{C_i}{\Theta(\log^{1/3} n)}$ with probability at least $1-K_i/n^2$. Thus, for $\mathcal{E}$ it holds that:
    
    $$\Prob{\neg\mathcal{E}} \leq \sum_i \Prob{\neg \mathcal{O}_i}+\Prob{\neg \mathcal{F}_i}+\Prob{\neg \mathcal{G}_i} \leq 3\sum_i \frac{K_i}{n^2} < \frac{1}{n},$$
    for sufficiently large $n$.
\end{proof}

\subsection{Proof of Lemma~\ref{lem:algcor}}

\begin{proof}
    Let a block have coordinates $(v_1,\dots,v_d)$, and let $n_i$ denote the number of blocks along dimension $i$. The traversal proceeds dimension by dimension in a nested fashion: to advance one step in dimension $i$, the algorithm requires a complete traversal of the $(i-1)$ preceding dimensions. Thus, the position of $(v_1,\dots,v_d)$ in the order can be expressed explicitly as
    \[
        o(v_1,\dots,v_d) \;=\; 1 + \sum_{i=1}^d \Big( \sigma_i(v_{i+1},\dots,v_d) \cdot (v_i-1) + (1-\sigma_i(v_{i+1},\dots,v_d)) \cdot (n_i-v_i) \Big) \cdot \prod_{j < i} n_j ,
    \]
    where $\sigma_i(\cdot) \in \{0,1\}$ encodes the direction of traversal in dimension $i$, which depends only on the parity of coordinates in higher dimensions. Since each $n_i$ is a power of two, this direction-flipping rule is consistent throughout the traversal.

    For adjacency, note that the algorithm changes exactly one coordinate of $v$ by $\pm 1$ at each step, never moving outside the boundary of $[n_1]\times\cdots\times[n_d]$. Hence consecutive blocks share a face (or hyperface). The only exception is the first and last blocks, which are opposite corners, and thus not consecutive in the traversal. Therefore, \Cref{alg:findorder} is correct.
\end{proof}

\subsection{Proof of Lemma~\ref{lem:OPT}}

\begin{proof}
By \Cref{conc}, the random variable $\OPT$ is sharply concentrated around its mean $\mathbb{E}[\OPT]$. By \Cref{exOPT}, we have $\mathbb{E}[\OPT] = (1+o(1))\beta_d n^{1-1/d}$. Setting $t = \tfrac{1}{2}\beta_d n^{1-1/d}$ in \Cref{conc} yields

$$
    \Prob{\OPT \le \mathbb{E}[\OPT] - \tfrac{1}{2}\beta_d n^{1-1/d}}
    \;\le\; 2 \exp\!\left(-\frac{c\,(\tfrac{1}{2}\beta_d n^{1-1/d})^2}{n^{1-2/d}}\right).
$$
Simplifying gives

$$
    \Prob{\OPT \le \tfrac{1}{2}\beta_d n^{1-1/d}}
    \;\le\; 2 \exp(-c' d n)
$$
for a suitable constant $c' > 0$. This proves the claim.
\end{proof}

\subsection{Proof of Lemma~\ref{lem:connect-blocks}}

\begin{proof}
Since consecutive blocks are neighbors, the connecting edge spans at most twice the diagonal of a block. The diagonal of a block is bounded by its side length times $\sqrt{d}$. From the partition analysis, the side length is at most $4(\tfrac{\log^{2}n}{n})^{1/d}$. Hence the connecting edge has length at most $2\sqrt{d}\cdot 4\Bigl(\tfrac{\log^{2}n}{n}\Bigr)^{1/d}.$ This is at most $2\beta_d\, n^{-1/d}\, \log^{2}n$, completing the proof.
\end{proof}

\subsection{More than $\ell$ Dimensions}\label{ellgeqd}

In the previous part of this section we analyzed the case $d < \ell$. We now justify this restriction by showing that when $d \geq \ell$, even the trivial strategy of visiting the points in arrival order yields a tour whose length is within a constant factor of $\OPT$. Thus, no sophisticated partitioning is required in high dimensions.

\begin{lemma}
When $d \geq \ell$ and the points are drawn i.i.d.\ from the uniform distribution on $[0,1]^d$, the tour that visits the points in arrival order has length $O(\OPT)$ with probability at least $1-1/n$.
\end{lemma}

\begin{proof}
Let $C_n := \sum_{i=1}^{n} \|x_i - x_{i+1}\|,$ where $x_{n+1}=x_1,$ denote the cost of the arrival-order tour. Since the points are i.i.d., its expected cost is

$$
    \E{C_n} = n \cdot \mu_d, 
$$
where $\mu_d := \E{\|X-Y\|}$ for $X,Y \sim \mathcal{U}([0,1]^d)$ independent.  
Anderssen et al.~\cite{anderssen} showed that

$$
    \tfrac{\sqrt{d}}{3} \;\leq\; \mu_d \;\leq\; \sqrt{\tfrac{d}{6}} .
$$
Hence

$$
    n \cdot \tfrac{\sqrt{d}}{3} \;\leq\; \E{C_n} \;\leq\; n \cdot \sqrt{\tfrac{d}{6}} .
$$

\medskip\noindent
\emph{Concentration.}  
Changing a single point $x_i$ affects at most two terms of $C_n$, and each term is at most $\sqrt{d}$. Thus, the Lipschitz constant is $2\sqrt{d}$. McDiarmid’s inequality gives

$$
    \Prob{|C_n - \E{C_n}| > \varepsilon} \;\leq\; 2\exp\!\left(-\frac{\varepsilon^2}{2 d n}\right).
$$
Setting $\varepsilon = n \sqrt{d/6}$ shows that with probability at least $1-2e^{-n/12}$,

$$
    C_n \;\leq\; (1+\sqrt{3/2})\, \E{C_n} \;\leq\; O(n \sqrt{d}).
$$

\emph{Comparison to $\OPT$.}  
By \Cref{lem:OPT}, with probability at least $1-2e^{-c' d n}$ we have

$$
    \OPT \;\geq\; \tfrac{1}{2}\,\beta_d\, n^{1-1/d}.
$$

Therefore, with overall probability at least $1-(2e^{-n/12}+2e^{-c'dn})$,

$$
    \frac{C_n}{\OPT}
    \;\leq\; \frac{(1+\sqrt{3/2}) \, n\sqrt{d/6}}{(1/2)\,\beta_d\, n^{1-1/d}}
    \;=\; c \cdot n^{1/d}.
$$

Since $d \geq \ell$, this ratio is $O(n^{1/\ell}) = O(1)$, which proves the lemma.
\end{proof}

\end{document}